\documentclass[12pt, aip, jmp, numerical]{revtex4-1}

\usepackage{amsmath,amsthm,amstext,amscd,amssymb,euscript,mathrsfs}
\usepackage{calrsfs,natbib}
\usepackage{epsfig}
\usepackage[colorlinks=true,linkcolor=blue]{hyperref}
\usepackage{amsmath,dsfont}
\usepackage[usenames]{color}
  
  \definecolor{grey}{gray}{0.6}

\long\def\drop#1{}

\newcommand{\epsi}{\ensuremath{\epsilon}}
\newcommand{\pee}{\ensuremath{\mathbb{P}}}
\newcommand{\be}{\begin{equation}}
\newcommand{\ee}{\end{equation}}
\newcommand{\caM}{{\mathcal M}}
\newcommand{\beq}{\begin{eqnarray}}
\newcommand{\eeq}{\end{eqnarray}}
\newcommand{\loc}{\mathcal{L}}
\def\R{\mathbb R}
\def\Z{\mathbb Z}
\def\E{\mathbb E}

\def\LBMP{L^{\mathrm{BMP}}}
\def\div{\mathop{\mathrm{div}}}

\def\M{\mathcal M}
\let\weakto\rightharpoonup
\def\Indicator{\mathds{1}}
\DeclareMathOperator\wlim{w-lim}
\newtheorem{theorem}{Theorem}
\newtheorem{lemma}[theorem]{Lemma}

\newcommand\T{\mathbb T}
\begin{document}
\title{Large Deviations in Stochastic Heat-Conduction Processes Provide a Gradient-Flow Structure for Heat Conduction}
\author{Mark A. Peletier}
\email{m.a.peletier@tue.nl}
\affiliation{Department of Mathematics and Computer Science and Institute for Complex Molecular Systems, Technische Universiteit Eindhoven, Postbus 513, 5600 MB Eindhoven, The Netherlands}
\author{Frank Redig}
\email{f.h.j.redig@tudelft.nl}
\affiliation{Delft Institute of Applied Mathematics,
Technische Universiteit Delft,  Mekelweg 4, 2628 CD Delft, The Netherlands}
\author{Kiamars Vafayi}
\email{k.vafayi@tue.nl}
\affiliation{Department of Mathematics and Computer Science, Technische Universiteit Eindhoven, Postbus 513, 5600 MB Eindhoven, The Netherlands}

\date{\today}

\begin{abstract}
We consider three one-dimensional continuous-time Markov processes on a lattice, each of which models the conduction of heat: the family of Brownian Energy Processes with parameter $m$, a Generalized Brownian Energy Process, and the Kipnis-Marchioro-Presutti process.
The hydrodynamic limit of each of these three processes is a parabolic equation, the linear heat equation in the case of the BEP$(m)$ and the KMP, and a nonlinear heat equation for the GBEP($a$). We prove the hydrodynamic limit rigorously for the BEP$(m)$, and give a formal derivation for the GBEP($a$).

We then formally derive the pathwise large-deviation rate functional for the empirical measure of the three processes. These rate functionals imply gradient-flow structures for the limiting linear and nonlinear heat equations. We contrast these gradient-flow structures with those for processes describing the diffusion of mass, most importantly the class of Wasserstein gradient-flow systems. The linear and nonlinear heat-equation gradient-flow structures are each driven by entropy terms of the form $-\log \rho$; they involve dissipation or mobility terms of order $\rho^2$ for the linear heat equation, and a nonlinear function of $\rho$ for the nonlinear heat equation.
\end{abstract}

\maketitle

\section{The Heat Equation}

The equation
\begin{equation}
\label{eq:heat}
\partial_t \rho = \Delta \rho
\end{equation}
is known both as the `diffusion equation' and as the `heat equation'. The two names reflect the fact that one arrives at the same equation via two completely different modeling routes---one describing the motion of particles, the other the conduction of heat.

Recently, one of these modeling routes has received special attention. Inspired by the observation that equation~\eqref{eq:heat} is a Wasserstein gradient flow (we describe this below), and that the Wasserstein gradient flows have many interesting and useful properties, we asked the question `how is the Wasserstein gradient-flow structure of~\eqref{eq:heat} related to the modeling of~\eqref{eq:heat}?'
We answered this question~\cite{AdamsDirrPeletierZimmer11,AdamsDirrPeletierZimmer12TR,DuongLaschosRenger13} by connecting the Wasserstein gradient-flow structure to the large-deviation behavior of stochastic particle systems.

Before turning to the content of the present paper, we briefly describe these previous results. A \emph{gradient flow} on a linear space $X$ is formally defined by a functional $E:X\to\R$ and a symmetric  operator $K:X'\to X$ (the so-called \emph{Onsager operator})
where $X'$ denotes the dual space of $X$. The equation~\eqref{eq:heat} has many such gradient-flow structures, each corresponding to a different triple $(X,E,K)$. One example is $X=X'=L^2(\R^n)$, $E(\rho) = \frac12\int_{\R^n} |\nabla \rho|^2$, and $K\xi = \xi$; then by definition
\[
\partial_t \rho = -K\Bigl(\frac{\delta E}{\delta \rho}(\rho)\Bigr) = -K (-\Delta \rho) = \Delta \rho.
\]
Note that we specify $K$ by describing its effect on the variational derivative $\delta E/\delta \rho$ of $E$. Other examples are
\begin{alignat}3
&X =H^{-1}(\R^n), &\qquad & E(\rho) = \frac12 \int_{\R^n} \rho^2, &\qquad & K\xi = -\Delta\xi,
\notag\\
&X = L^1_{\geq0}(\R^n), && E(\rho) = \int_{\R^n} \rho\log \rho, && K_\rho\xi = -\div \rho\nabla \xi,\label{ex:WassersteinGF}\\
&X = L^1(\R^n), && E(\rho) = \int_{\R^n} f(\rho), && K_\rho\xi = -\div \Bigl(\frac1{f''(\rho)}\nabla \xi\Bigr).
\label{ex:general}
\end{alignat}
Each of these gradient-flow structures leads to the same equation~\eqref{eq:heat}.
The last two examples show that $K$ may depend on $\rho\in X$; in geometric terms, $K$ maps the cotangent bundle to the tangent bundle, and therefore may depend on position.

Example~\eqref{ex:WassersteinGF} is the Wasserstein gradient-flow structure that was mentioned above. We showed in~\onlinecite{AdamsDirrPeletierZimmer11,AdamsDirrPeletierZimmer12TR} how the Wasserstein structures of this equation and of more general convection-diffusion equations  arise from the properties of an underlying stochastic process: systems of particles whose positions are described by a stochastic differential equation. The large deviations of empirical measures of these particles directly identify a gradient-flow structure, which is of the form~\eqref{ex:WassersteinGF}. This identification consists of two parts:
\begin{enumerate}
\item The equilibrium large-deviation rate functional of the \emph{invariant measure} is the driving functional $E$;
\item The non-equilibrium large-deviation rate functional of the \emph{law of the time-courses of the empirical measure} identifies the operator $K$.
\end{enumerate}
Through these identifications, the properties of this gradient-flow structure can be traced back to the properties of the underlying stochastic process. For instance, the presence or absence of interaction only appears in the form of the driving functional~$E$, and the form of $K$ is completely determined by the Brownian noise in the SDE. As a result, the Wasserstein gradient-flow structure of~\eqref{eq:heat} can now be completely explained in terms of properties of such underlying models of diffusion. The same arguments apply to various related equations~\cite{PeletierRengerVeneroni13,DuongPeletierZimmer12TR,DuongPeletierZimmer13,Renger13TH}, and we have recently showed how it is shared among all processes with detailed balance~\cite{MielkeRengerPeletier13TR}. The insight that this provides is of significant use in the modeling of other real-world systems.

However, this explanation only applies to \emph{diffusion}---to the derivation of~\eqref{eq:heat} in which $\rho$ is a \emph{density of particles} that evolves by thermal agitation. For the other interpretation, where $\rho$ is a \emph{temperature} or a \emph{heat content}, and where heat disperses through heat conduction, very little is known in this direction. Is there a corresponding gradient-flow structure of~\eqref{eq:heat}, that derives from large deviations of some stochastic process? If so, is it the same as for diffusion, or a different one?

\medskip

In this paper we investigate this question. We study stochastic processes that mimic the \emph{conduction of heat:} in these processes a local `internal' energy is exchanged through random nearest-neighbor interactions. The hydrodynamic or continuum limit of these systems is again a parabolic partial differential equation---and in some of the cases it is exactly the linear heat equation~\eqref{eq:heat}. As in the case of diffusion, the large deviations away from this continuum limit identify a natural gradient-flow structure, that we discuss in detail below.

In fact, we will consider three stochastic processes:
\begin{enumerate}
\item The \emph{Brownian Energy Process} with parameter $m$ (BEP($m$), where $m$ is a positive real number);
\item A \emph{Generalized} Brownian Energy Process with parameter $a$ (GBEP($a$), where $a$ is a positive function);
\item The \emph{Kipnis-Marchioro-Presutti} process (KMP).
\end{enumerate}
Each of these describes heat conduction in a slightly different way, and we will use the similarities and differences to investigate how different aspects of these stochastic processes lead to different continuum limits, and to different corresponding gradient-flow structures. This will allow us to link the properties of the gradient-flow ingredients to specific aspects of the underlying stochastic processes.

\medskip

The outline of this paper is as follows. In Section~\ref{sec:defs-of-processes} we describe the three processes mentioned above. Most of the further development is done for the BEP($m$): in Sections~\ref{sec:eq-props-BEP} and~\ref{sec:eq-LDP} we study the equilibrium properties of the BEP, and in Sections~\ref{sec:HDL}, \ref{sec:WABEP}, and~\ref{sec:LDP} the hydrodynamic limit  and the corresponding large deviation principle. In Section~\ref{sec:otherprocesses} we discuss how the corresponding properties of the two other processes can be derived. Finally, in Section~\ref{sec:GF} we discuss the gradient-flow structures that these large deviations suggest, and in Section~\ref{sec:discussion} the implications of these structures.

\section{The three stochastic processes: BEP($m$), GBEP($a$), and KMP}
\label{sec:defs-of-processes}

Heat conduction is the transfer of internal energy between neighboring parts of a material. In the most straightforward  microscopic models of heat conduction, matter is represented as a collection of oscillators on lattice points; each oscillator is a little isolated Hamiltonian system, such as  a little mass-spring system. If there is no interaction, then the oscillator has a finite energy---the Hamiltonian---that is conserved over time.

Heat conduction arises when one introduces interaction between the oscillators, either deterministically or stochastically, allowing energy exchange between neighboring sites. Deterministic interaction between neighboring oscillators is natural, since both classical and quantum descriptions of reality are Hamiltonian; however, the difficulties surrounding the analysis of deterministic heat-conducting systems are formidable (see e.g.~\onlinecite{BonettoLebowitzRey-Bellet00,Dhar08}). For this reason the interaction is often chosen to be stochastic, and many varieties exist in addition to the ones we study here; see e.g.~\onlinecite{KipnisMarchioroPresutti82,BernardinOlla05,BasileBernardinOlla06,Bernardin07,BricmontKupiainen07,Bernardin08,BasileBernardinOlla09,BernardinOlla11}.

The models that we study in this paper are all relatives of a single, prototypical model, that we now describe. Afterwards, we introduce the three models which are the subject of this paper and explain how they are related.

\subsection{The Brownian Momentum Process}

\label{sec:BMP}

In the \emph{Brownian Momentum Process}, physical space is represented by a one-dimensional lattice of $N$ points $i=1,\dots,N$; we silently assume periodicity modulo $N$, i.e.\ the point $i=N+1$ is the same as the point $i=1$. The \emph{state} of the BMP process is a vector of \emph{momenta} $p_i$, $i=1,\dots,N$. The BMP process is characterized by its generator $\LBMP$ that acts on smooth functions $f:\R^N\to\R$, given by
\begin{equation}
\label{def:L-BMP}
\LBMP = \sum_{i=1}^N \LBMP_{i,i+1},
\qquad \LBMP_{i,i+1}f = \frac12 \Bigl(p_i\frac\partial{\partial p_{i+1}}-p_{i+1}\frac\partial{\partial p_i}\Bigr)^2 f.
\end{equation}
Equivalently, the $p_i$ solve the It\^o stochastic differential equation
\begin{equation}
\label{eq:sde-pi}
dp_i = p_{i+1}dB_{i,i+1} - p_{i-1} dB_{i-1,i} - p_i dt,
\end{equation}
where the $B_{i,i+1}$ are independent Brownian motions.

The pairwise generator $\LBMP_{i,i+1}$ can be interpreted as a Brownian motion on the circle given by the condition $p_i^2 + p_{i+1}^2 = \mathrm{constant}$. This follows from noting that  $\LBMP_{i,i+1}$ preserves $p_i^2 + p_{i+1}^2$, and in terms of the angle coordinate $\varphi = \arctan (p_{i+1}/p_i)$ we can write
$\LBMP_{i,i+1} = \tfrac12 \partial^2/\partial\varphi^2$, representing a simple Brownian motion in $\varphi$. In terms of the SDE~\eqref{eq:sde-pi}, the conservation of $p_i^2+p_{i+1}^2$ can be observed in the double occurrence of the noise $dB_{i,i+1}$, once in the equation for $p_i$ (with prefactor $p_{i+1}$) and once in the equation for $p_{i+1}$, with the opposite sign, and with prefactor $p_i$. The term $-p_i dt$ is an It\^o correction that compensates for the asymmetry in the definition of an It\^o SDE.

This process is a model of heat conduction in the following sense.
Up to a constant, the sum $p_i^2 + p_{i+1}^2$ should be interpreted as the total kinetic energy of the two sites $i$ and $i+1$. Therefore each neighbour interaction, characterized by $\LBMP_{i,i+1}$,  is a Brownian motion over all possible distributions of the preserved kinetic energy over the two momenta $p_i$ and $p_{i+1}$, and therefore implements a stochastic exchange of energy between neighbouring sites.

\subsection{The Brownian Energy Process with parameter $m$}

Because of our focus on heat conduction, we shall be interested in tracking the squares $z_i := p_i^2$, i.e.\ the site energies, rather than in the momenta $p_i$ themselves. The \emph{Brownian Energy Process} (BEP), with parameter \emph{one}, is the process that governs the evolution of the vector of energies $(z_1,\dots,z_N)$. It happens to be a Markov process itself: the evolution of the $z_i$ can be characterized  in terms of the $z_i$ themselves, with generator
\begin{equation*}
L^{\mathrm{BEP}}=\sum_{i=1}^{N} L^{\mathrm{BEP}}_{i,i+1},
\end{equation*}
where
\begin{equation}\label{bepgen}
L^{\mathrm{BEP}}_{i,i+1} =  2z_i z_{i+1}
\left(\partial_i-\partial_{i+1}\right)^2 -(z_i-z_{i+1})\left(\partial_i-\partial_{i+1}\right).
\end{equation}
Here we write $\partial_i := {\partial}/{\partial z_i}$ for the derivative with respect to the  variable $z_i$. The corresponding SDE is
\begin{equation}
\label{eq:sde-zi}
dz_i = 2\sqrt{z_iz_{i+1}} \, dB_{i,i+1} - 2\sqrt{z_{i-1}z_i}\, dB_{i-1,i} + (z_{i+1}-2z_i+z_{i-1})dt.
\end{equation}
This process is the Brownian Energy Process---with parameter one.

\medskip
The \emph{Brownian Energy Process with parameter $m$}, BEP($m$), is a generalization of this process, in which not a single scalar quantity $p_i$, but $m$ scalar quantities $p_i^j$, $j=1,\dots m$ are defined at each lattice point $i$. We define the generator for this $m$-fold BMP by imposing the same interaction between all $m\times m$ neighboring pairs, i.e.
\[
L^{\mathrm{BMP}(m)} = \sum_{i=1}^N L^{\mathrm{BMP}(m)}_{i,i+1},\qquad
L^{\mathrm{BMP}(m)}_{i,i+1} = \sum_{j,j' = 1}^m \frac12 \Bigl(p_i^j\frac\partial{\partial p_{i+1}^{j'}}-p_{i+1}^{j'}\frac\partial{\partial p_i^j}\Bigr)^2.
\]
The total energy for each site $i$ is now defined as $z_i := \sum_{j=1}^m (p_i^j)^2$, and this is again a Markov process with generator
\begin{equation}
\label{def:L-BEPm}
L^{\mathrm{BEP}(m)} = \sum_{i=1}^N L^{\mathrm{BEP}(m)}_{i,i+1},\qquad
L^{\mathrm{BEP}(m)}_{i,i+1} = 2z_iz_{i+1}\left(\partial_i-\partial_{i+1}\right)^2 -m(z_i-z_{i+1})\left(\partial_i-\partial_{i+1}\right).
\end{equation}
Note the small difference with the generator~\eqref{bepgen} of the BEP(1): only the coefficient of the drift term is $m$-dependent. Also note that for the definition of the process $z_i$, $m$ need not be integer---any $m>0$ is admissible, although the connection to an underlying Momentum Process of course only exists for integer $m$.

\medskip
The Brownian Momentum Process was introduced in~\onlinecite{GiardinaKurchanRedig07} with imposed-temperature boundary conditions, and further studied in~\onlinecite{GiardinaKurchanRedigVafayi09,GiardinaRedigVafayi10,RedigVafayi11}.
In~\onlinecite{RedigVafayi11} we studied non-equilibrium stationary states of this system, generated by imposing different temperatures at the two ends of a one-dimensional system. We showed, among other things, that the expectation of the temperature along these non-equilibrium steady states is linear. This is of course consistent with non-equilibrium steady states of the macroscopic equation~\eqref{eq:heat}. The Brownian Energy Process was introduced in~\onlinecite{GiardinaKurchanRedigVafayi09} and further studied in \onlinecite{GiardinaRedigVafayi10,GrosskinskyRedigVafayi11}.

\subsection{A Generalized Brownian Energy Process with parameter $a$}

For reasons that will become clear below, it is interesting to consider a different generalization of the BEP, in which the noise terms in~\eqref{eq:sde-pi} and~\eqref{eq:sde-zi} are modified in an energy-conserving manner, as follows. Choose a continuous function $a:[0,\infty)\to(0,\infty)$, and define
\[
a_{i,i+1}:= a\Bigl(\tfrac12(z_i+z_{i+1})\Bigr).
\]
The \emph{Generalized Brownian Momentum Process} is the SDE
\begin{equation}
\label{def:GBMP}
dp_i = p_{i+1}a_{i,i+1}\, dB_{i,i+1} - p_{i-1} a_{i-1,i}\, dB_{i-1,i} - \frac12 p_i(a_{i,i+1}^2 + a_{i-1,i}^2) \,dt,
\end{equation}
and the corresponding \emph{Generalized Brownian Energy Process} with parameter $a$ is the SDE
\begin{equation}
\label{def:GBEP}
dz_i = 2a_{i,i+1}\sqrt{z_iz_{i+1}} \, dB_{i,i+1} - 2a_{i-1,i}\sqrt{z_{i-1}z_i}\, dB_{i-1,i} +  \bigl((z_{i+1}-z_i)a_{i,i+1}^2 + (z_i-z_{i-1})a_{i-1,1}^2\bigr)\,dt,
\end{equation}
with generator
\[
L^{\mathrm{GBEP}(a)} :=\sum_{i=1}^N a_{i,i+1}^2\biggl[ 2z_iz_{i+1} (\partial_i-\partial_{i+1})^2
  + (z_{i+1}-z_i)(\partial_i-\partial_{i+1})\biggr]
  = \sum_{i=1}^N a_{i,i+1}^2 L^{\mathrm{BEP}}_{i,i+1}.
\]
As far as we are aware, this process has not been discussed in the literature.

\subsection{The Kipnis-Marchioro-Presutti process}

In~\onlinecite{KipnisMarchioroPresutti82} Kipnis, Marchioro, and Presutti introduced  a model for heat conduction. In the KMP process energy is not exchanged continuously, as in the processes described above, but discretely: at exponentially distributed times we choose a random pair of neighbors, say with energies $z_i$ and $z_{i+1}$; we then equidistribute the joint energy $z_i+z_{i+1}$ over the two sites, with a fraction $s$ on one site and a fraction $1-s$ on the site. Here $s$ is drawn from the uniform distribution on $[0,1]$. The KMP process has the generator\cite{Note1}
\[
(L^{\mathrm{KMP}} f)(z) := 2\sum_{i=1}^N \int_0^1 \bigl[f(z^{i,i+1,s})-f(z)]\, ds,
\]
where the modified $z$ is defined as
\[
z^{i,i+1,s}_j := \begin{cases}
s(z_i+z_{i+1}) & \text{if } j=i,\\
(1-s)(z_i+z_{i+1}) & \text{if } j=i+1,\\
z_j &\text{otherwise}.
\end{cases}
\]

As pointed out in~\onlinecite{GiardinaKurchanRedigVafayi09}, the KMP process is related to the process BEP(2), in the following way. Choose any $i$ and consider the {single-pair generator} $L^{\mathrm{BEP}(m)}_{i,i+1}$. This single-pair generator has a unique invariant measure for every $m$, which describes the equilibrium distribution of energy over the two sites $i$ and $i+1$, assuming all other $z_j$ are fixed. For $m=2$, this measure is exactly the uniform distribution of energy between the two sites $i$ and $i+1$. (The general case is that if $(z_i,z_{i+1}) = (E-e,E+e)$, then $e\in[0,E]$ has probability density $\sim(E^2-e^2)^{m/2-1}$.)

Therefore $L^{\mathrm{KMP}}$ can be considered an `instantaneously equilibrated' or `instantaneously thermalised'~\cite{GiardinaKurchanRedigVafayi09} version of $L^{\mathrm{BEP}(2)}$---as if the exchange between $i$ and $i+1$ is not smeared out over time but concentrated into very short bursts, during which, complete equilibration within that specific bond takes place.

\subsection{Macroscopic quantities}

In each of the three processes BEP($m$), GBEP($a$), and KMP, the variable $z_i$ at site $i$ represents the local or microscopic energy at that site. The natural way to connect this local energy with macroscopic quantities such as temperature and macroscopic energy is as follows.
Given an energy state $Z = \{z_1,z_2,...,z_N\}$, we define the empirical energy measure $\pi_N(Z)$ to be the probability measure
\begin{equation}
\label{empm}
\left[\pi_N (Z)\right](dx):=
\frac{1}{N} \sum_{i=1}^{N} z_i \:
\delta_{{i}/{N}} (dx),
\end{equation}
where $\delta_{i/N}(x)$ is the Delta measure at $x=i/N$. This makes $\pi_N$ a probability measure on the flat one-dimensional torus $\T = \R/\Z$, i.e.\ the interval $[0,1]$ with periodic boundary conditions.

We show below how, in the limit $N\to\infty$, $\pi_N$ converges weakly to a deterministic limit profile $\rho$, both in equilibrium and in non-equilibrium situations. The limit $\rho$ will be interpreted as macroscopic energy. The large deviations away from the deterministic limit will provide us with the structure of the corresponding gradient flow.

\section{Equilibrium properties of the BEP($m$)}
\label{sec:eq-props-BEP}

In the next few sections we focus on the BEP($m$); we study the hydrodynamic limit and the large deviations of the model, and we derive from the large deviations the corresponding gradient-flow structure. In Section~\ref{sec:otherprocesses} we describe the differences for the two other processes.

\medskip

We first consider equilibrium properties.
Since $L^{\mathrm{BEP}(m)}\bigl[\sum_i z_i\bigr]= 0$, the BEP($m$) deterministically conserves the total energy $\sum_i z_i$---this property is also natural from the construction of the process. As a consequence the process does not have a single invariant measure, but a large collection of invariant measures, generated by a particular one-parameter family of invariant measures, indexed by a parameter $\theta>0$. For each~$\theta$,  with the interpretation of temperature, this invariant measure is a product measure over the sites $i$:
\begin{equation}
\label{def:nu}
\nu_{N, \theta,m} (dZ)=\prod_{i=1}^N \nu_{\theta,m} (dz_i),
\qquad\text{where}\qquad
\nu_{\theta,m} (dz) = \frac{1}{\theta^{m/2} \Gamma \left(\frac{m}{2}\right)}z^{-1+{m}/{2} } e^{-z/\theta}\, dz,
\end{equation}
and $\Gamma$ is the Gamma-function
\begin{equation*}
\Gamma (k) := \int_0^{\infty} e^{-t} t^{k-1} dt.
\end{equation*}
The process is reversible with respect to each of these invariant measures.

In equilibrium, since this invariant measure is a product measure, to every parameter $\theta$ corresponds a uniform (i.e., $i$-independent) average energy density $\rho_0$ given by
\begin{equation}
\label{def:rho0}
\rho_0(\theta)=
\E_{\nu_{\theta,m}} [ z_i ]= \frac{1}{\theta^{m/2} \Gamma \left(\frac{m}{2}\right)} \int_0^{\infty} z^{{m}/{2} } e^{-z/\theta}dz=\frac{m\theta}{2}.
\end{equation}
Since there is a one-to-one correspondence between $\theta$ and $\rho_0$, we can equivalently index the empirical measure by $\rho_0$ as $\nu_{N, \rho_0,m}=\nu_{N, \theta(\rho_0),m}$ instead.

Moreover, if $Z$ is distributed according to $\nu_{N, \rho_0,m}$, the empirical energy $\pi_N (Z)$ converges in probability to
the constant measure $\rho_0$, by the law of large numbers:
for each $\delta>0$ and each
test function $\phi\in C_b(\T)$ we have
\begin{equation}
\lim_{N\rightarrow \infty}
\nu_{N,\rho_0,m}\Big( \big|
\langle \pi_N(Z ),\phi\rangle - \langle \rho_0,\phi\rangle
\big| > \delta \Big)=0
\label{llnr}
\end{equation}
where
\[
\langle \rho_0,\phi\rangle:=\int_\T \rho_0(x)\phi(x)\ dx.
\]

\section{Equilibrium Large Deviations of the BEP}
\label{sec:eq-LDP}

We already mentioned in the introduction that for reversible processes, such as the three processes of this paper, the large-deviations rate functional of the invariant measure is the functional that drives the gradient-flow evolution~\cite{MielkeRengerPeletier13TR}. We now calculate this rate functional.

A large-deviation principle states that the probability that the empirical energy $\pi_N(Z)$ deviates considerably from its most likely value $\rho_0$
and thus becomes close to some other profile $\rho$ is
exponentially small in $N$ for large $N$ and is determined by a rate functional $S_{\rho_0,m}$ according to
\begin{equation*}
\nu_{N,\rho_0,m}\left(\pi_N(Z)\approx \rho\right)\sim
\exp\bigl\{-N \, S_{\rho_0,m}(\rho) \bigr\}.
\end{equation*}
More precisely,
\begin{theorem}
For all Borel subsets $A$ of the set of non-negative measures $\M_+(\T)$,
\begin{align}
\label{ldp:invmeasure}
-\inf_{\mathrm{int}\, A}S_{\rho_0,m}&\leq \liminf_{N\to\infty} \frac1N \log \nu_{N,\rho_0,m}(A)
\leq \limsup_{N\to\infty} \frac1N \log \nu_{N,\rho_0,m}(A) \leq -\inf_{\mathrm{cl}\,A}S_{\rho_0,m},
\end{align}
where
\begin{equation}
\label{def:S0}
S_{\rho_0,m} (\rho)
= \frac{m}{2}\int_0^1  \left( \frac{\rho(x)}{\rho_0}-1-\log  \frac{\rho(x)}{\rho_0} \right)\, dx.
\end{equation}
\end{theorem}

The existence of a functional $S_{\rho_0,m}$ satisfying~\eqref{ldp:invmeasure} and its explicit form both follow from the G\"artner-Ellis theorem~\onlinecite[Th.~4.5.20]{DemboZeitouni98}; the proof is standard, and we sketch only those details that we need later on.

\begin{proof}
We first calculate $S_{\rho_0,m}$. Set $\theta = 2\rho_0/m$ (following~\eqref{def:rho0}) and define the logarithmic cumulant-generating function as
\begin{equation}
\label{press}
G : C_b(\T) \to\R,\qquad
G(\phi) :=
\lim_{N\rightarrow \infty} \frac{1}{N}
\log \E_{\nu_{N,\theta,m}^{}}
\Big( e^{N \langle \phi,\pi_{N}(Z )\rangle} \Big).
\end{equation}
From the definition of the energy profile we have, for continuous $\phi\in C_b(\T)$,
\begin{equation*}
N \langle \phi,\pi_{N}(Z )\rangle = \sum_{i=1}^{N} \phi_i z_i, \qquad \text{where we write }\phi_i:=\phi(i/N).
\end{equation*}
Due to the independence of different sites in the product measure we can use the single-site generating function
\begin{equation}
\label{calc:single-site-gen-func}
\E_{
\nu_{\theta,m}} [ e^{\phi_i z_i} ]= \frac{1}{\theta^{m/2} \Gamma \left(\frac{m}{2}\right)} \int_0^{\infty} z^{-1+{m}/{2}} e^{-z(-\phi_i+1/\theta)}dz=
\left(\frac{1}{1-\theta \phi_i}\right)^{m/2}
\end{equation}
to calculate
\begin{alignat*}2
G(\phi) &=
\lim_{N\rightarrow \infty} \frac{1}{N}
\log \E_{\nu_{N,\theta,m}}
\Big( \prod_{i} e^{ \phi_i z_i} \Big)
\\
&=\lim_{N\rightarrow \infty} \frac{1}{N}
\log
\Big( \prod_{i} \E_{\nu_{\theta,m}} e^{ \phi_i z_i} \Big)
&\qquad&\text{by independence}
\\
& =
\lim_{N\rightarrow \infty} \frac{1}{N}
\log
 \prod_{i} \left(\frac{1}{1-\theta \phi_i}\right)^{m/2}
&\qquad&\text{by }\eqref{calc:single-site-gen-func}
\\
& =-\frac{m}{2}
\lim_{N\rightarrow \infty} \frac{1}{N}
\sum_{i} \log
  \left({1-\theta \phi_i}\right)
\\
& = -\frac{m}{2}
\int_\T   \log
  \left({1-\theta \phi(x)}\right)\, dx .
\end{alignat*}
(Note how the local nature of the functional $G$ (the fact that $G$ is a simple integral over $\T$) arises from the product  structure of the measure $\nu_{N,\theta,m}$.)
We then obtain the formula~\eqref{def:S0} for $S_{\rho_0,m}$ as the Legendre transform of $G$ by a direct calculation.

After this calculation the proof of the property~\eqref{ldp:invmeasure} follows from~\onlinecite[Th.~4.5.20(c)]{DemboZeitouni98}. There is one non-trivial condition to be verified, which is the density of an appropriate class of exposed points of $\M_+(\T)$; this can be ascertained by noting that the subset of probability measures with Lebesgue density that is continuous and bounded away from zero is dense in $\M_+(\T)$ with the narrow topology and satisfies the remaining requirements.
\end{proof}

\section{Hydrodynamic limit of the time evolution}
\label{sec:HDL}

We now turn to the time evolution of the BEP($m$). First we note that the natural time scale for the empirical measure $\pi_N(Z)$ is $N^2$, which leads us to define the rescaled (accelerated) generator
\begin{equation*}
\label{gen}
L_N^{\mathrm{BEP}(m)} :=   N^2  L^{\mathrm{BEP}(m)} .
\end{equation*}
Next we describe the initial state. Choose a fixed measure $\rho^0\in \M_+(\T)$, and choose a sequence of vectors $Z^0_N = (z_{N,1}^0,\dots,z^0_{N,N})\in \R^N$ such that $\pi_N(Z_N^0) \weakto \rho^0$ as $N\to\infty$. Moreover we assume that $z_{N,i}^0$ are bounded uniformly in $N$. By ${\mathbb P}^{\mathrm{BEP}(m)}_N$ we denote the distribution of the BEP($m$)
process, accelerated by a factor $N^2$---i.e.\ with generator $L^{\mathrm{BEP}(m)}_N$---where the initial condition is equal to $Z_N^0$.
$\E^{\mathrm{BEP}(m)}_N$ is the expectation with respect to  ${\mathbb P}^{\mathrm{BEP}(m)}_N$. From now on we also consider the process $Z$ itself to be accelerated by a factor $N^2$, i.e. $Z$ has law~$\mathbb P^{\mathrm{BEP}(m)}_N$.

The limit of the process will be the unique solution $\rho$ of the linear heat equation,
\begin{equation}
\label{eq:limitBEP}
\partial_t \rho(t,x)
=m\, \partial_{xx} \rho(t,x),\qquad\text{for }x\in \T\text{ and }t>0.
\end{equation}
with initial datum
\[
\rho(0,\cdot) = \rho^0.
\]

The following theorem specifies the exact sense in which $\pi_N(Z)$ converges to $\rho$.
\begin{theorem}
\label{th:hydlimit}
For all $\phi\in C_b(\T)$ and $t>0$,
\begin{equation}
\lim _{N\to \infty}{\mathbb P}^{\mathrm{BEP}(m)}_N
\Big( \big| \langle\pi_N(Z(t)),\phi\rangle
- \langle \rho(t),\phi \rangle \big| > \delta\Big)=0.
\end{equation}
Also a stronger property holds: the whole random trajectory $\{ \pi_N(Z(t)): 0\leq t\leq T\}$
converges {weakly to the Dirac measure concentrated on the trajectory} $\{ \rho(x,t)\, dx: 0\leq t\leq T\}$, in the {Skorokhod}
topology of {$D([0,T];\caM_+(\T))$}. Here $\caM_+(\T)$ is the space of non-negative measures on $\T$ with the topology of weak convergence.
\end{theorem}

A {sketch of proof} is as follows. If $\phi$ is a smooth function on $\T$,
we have from the definition of the  generator of Markov processes
\begin{equation*}
\frac{d}{dt}  \E^{\mathrm{BEP}(m)}_N
\big[  \langle\pi_N(Z(t)),\phi\rangle\big]
= \E^{\mathrm{BEP}(m)}_N\big( L_N^{\mathrm{BEP}(m)} \langle\pi_N(Z(t)),\phi \rangle\big).
\end{equation*}
From the definition of empirical energy we have
\begin{equation*}
\langle \pi_{N}(Z(t)) , \phi \rangle = \frac{1}{N} \sum_{i=1}^{N} z_i(t)  \phi_i,
\end{equation*}
where again we write $\phi_i:=\phi(i/N)$.
For fixed $i$, let $f_i$ be the coordinate function $f_i(Z) := z_i$.
Then
\begin{align*}
L^{\mathrm{BEP}(m)}f_i&=(L^{\mathrm{BEP}(m)}_{i-1,i}+L^{\mathrm{BEP}(m)}_{i,i+1})z_i
=\Bigl[2z_{i-1} z_{i}
\left(\partial_{i-1}-\partial_{i}\right)^2 -m (z_{i-1}-z_{i})\left(\partial_{i-1}-\partial_{i}\right)\\
&\quad \phantom{(L_{i-1,i}+L_{i,i+1})z_i
=}{}+
2z_i z_{i+1}
\left(\partial_i-\partial_{i+1}\right)^2 -m (z_i-z_{i+1})\left(\partial_i-\partial_{i+1}\right)\Bigr]z_i\\
&= m\bigl[ z_{i-1}-2 z_i+z_{i+1} \bigr],
\end{align*}
so that
\begin{equation*}
 L_N^{\mathrm{BEP}(m)} \langle\pi_N(Z(t)),\phi \rangle = N \sum_{i=1}^{N} (L^{\mathrm{BEP}(m)} f_i)(Z(t))  \phi_i
 =mN \sum_{i=1}^{N} \phi_i \big[ z_{i-1}-2 z_i+z_{i+1} \big].
\end{equation*}
We now apply summation by parts, using the periodicity of the lattice, to obtain
\begin{align*}
 L^{\mathrm{BEP}(m)}_N \langle\pi_N(Z(t)),\phi \rangle
 &=mN \sum_{i=1}^{N}  z_i \big[ \phi_{i-1}-2 \phi_i+\phi_{i+1} \big]\\
 &=\frac{m}{N} \sum_{i=1}^{N}  z_i \Delta_N \phi_i
\approx  m \langle \pi_N(Z(t)),\partial_{xx} \phi\rangle,
\end{align*}
where  $(\Delta_N\phi)_i:=N^2 \big[  \phi_{i-1}-2 \phi_i+ \phi_{i+1}\big] $ is a discrete approximation of the Laplacian and
in the last step we used the smoothness of $\phi$. Then
\[
\frac{d}{dt} \langle \pi_N(Z(t)),\phi\rangle \approx m \langle \pi_N(Z(t)),\partial_{xx} \phi\rangle,
\]
which is a weak formulation of~\eqref{eq:limitBEP}.

\bigskip

The rigorous proof of Theorem~\ref{th:hydlimit} follows from standard arguments for gradient processes that can be found in e.g.~\onlinecite[Ch.~4]{KipnisLandim99}. These revolve around the martingale
\begin{eqnarray}\label{marti}
M^N_t
&:=&
\langle\pi_N (Z(t)), \phi\rangle -\langle\pi_N (Z_N^0), \phi\rangle -\int_0^t  L_N^{\mathrm{BEP}(m)}\langle \pi_N(\cdot),\phi\rangle(Z(s)) \ ds
\nonumber\\
&=&
\langle\pi_N (Z(t)), \phi\rangle -\langle\pi_N (Z_N^0), \phi\rangle
 -\int_0^t  \langle \pi_N(Z(s)),m\partial_{xx} \phi\rangle\ ds+ o(1)
\end{eqnarray}
where $o(1)$ converges to zero uniformly as $N\to\infty$.

The essential property of $M^N$ that provides both compactness and convergence is the vanishing of the quadratic variation, the content of the following lemma.
\begin{lemma}
The quadratic variation of $M^N$ satisfies
\be\label{ooops}
\forall t>0, \qquad \lim_{N\to\infty}\E^{\mathrm{BEP}(m)}_N[M^N]_t=0.
\ee
\end{lemma}

\begin{proof}
The quadratic variation of $M^N_t$ is given by
\[
[ M^N]_t = \int_0^t \left(L_N^{\mathrm{BEP}(m)}(g_N^2)(Z(s))- 2g_N(Z(s)) L_N^{\mathrm{BEP}(m)}(g_N)(Z(s))\right)\, ds
\]
where
\[
g_N(Z(t)) = \frac{1}{N} \sum_{i=1}^{N} z_i(t)  \phi_i.
\]
Explicit computation then gives
\[
\langle M^N_t\rangle = \frac{4}{N^2} \sum_{i=1}^N \phi_i^2\int_0^t   u_i (Z(s)) ds
\]
with
\[
u_i(z) =  z_i z_{i+1} + z_{i-1}z_i.
\]

We use the assumption that the initial condition is bounded, i.e.,  that
$\sup_{i} z_{N,i}^0 <C$.
Then we have
\[
\E^{\mathrm{BEP}(m)}_N\biggl(\sum_{i=1}^N \int_0^t  \phi_i^2 u_i (Z(s)) \,ds\biggr)=\sum_{i=1}^N \phi_i^2 \int_0^t   \E^{\mathrm{BEP}(m)}_N(u_i (Z(s)))\, ds.
\]
Therefore, in order to obtain \eqref{ooops}, it suffices to see that
\[
\E^{\mathrm{BEP}(m)}_N (z_i(s) z_{i+1}(s))
\]
is bounded uniformly in $N$ and in $0\leq s\leq t$. The expression $z_i z_{i+1}$ is equal to $K_m D(\delta_i + \delta_{i+1}, z)$, where
$D(\xi,z)$ are the duality polynomials introduced in \onlinecite{GiardinaKurchanRedig07,GiardinaRedigVafayi10}, {and where
$K_m$ is a constant only depending on $m$}. Therefore, by duality of the BEP$(m)$ and the Symmetric Inclusion Process (SIP, see \onlinecite{GiardinaKurchanRedig07,GiardinaRedigVafayi10} for more details on this process and the duality), we have
\[
\E^{\mathrm{BEP}(m)}_N(z_i(s) z_{i+1}(s))= K_m\E_{i,i+1}^{\mathrm{SIP}}D(\delta_{X_t}+\delta_{Y_t}, z)\leq B_m (C+C^2),
\]
where $\E_{i,i+1}^{\mathrm{SIP}}$ denotes expectation in the SIP starting with a single particle at both the sites
$i$ and $i+1$, and where in the last step we used that $D(\delta_{X_t}+\delta_{Y_t}, z)$ is a second order polynomial in
the variables $z_i$, with coefficients only depending on $m$.
\end{proof}

From~\eqref{ooops}, using Doob's inequality, we obtain that for all test functions $\phi$
\begin{eqnarray*}
&&\sup_{0\leq t\leq T}\left|\langle\pi_N (Z(t)), \phi\rangle -\langle\pi_N (Z_N^0), \phi\rangle
 -\int_0^t  \langle \pi_N(Z(s)),m\,\partial_{xx} \phi\rangle\ ds\right|
\end{eqnarray*}
converges to zero in probability, which is exactly the statement of the hydrodynamic limit in the path space
version.

\section{The Weakly Asymmetric BEP($m$)}

\label{sec:WABEP}
For the calculation of the nonequilibrium large deviations of the BEP($m$) we need to calculate the hydrodynamic limit of a weakly-perturbed version of the BEP($m$), the Weakly Asymmetric BEP($m$) or WABEP($m$):
\begin{equation*}
L_N^{\mathrm{WABEP}(m)}:=L_N^{\mathrm{BEP}(m)}+\Gamma.
\end{equation*}
There are various ways to introduce this asymmetry. From group-theoretic considerations in~\onlinecite{GiardinaKurchanRedig07,GiardinaRedigVafayi10} we choose $\Gamma$ to be
\begin{equation*}
\Gamma=-N^2\sum_{i=1}^{N-1} E_i z_i z_{i+1} \left( \partial_i - \partial_{i+1} \right),
\end{equation*}
where the $E_i$ are to scale as $1/N$ as $N\to\infty$. One way to achieve this, which will also simplify formulas later, is to choose
\begin{equation*}
E_i:=H_{i+1}-H_i,
\end{equation*}
where $H_i:= H_i(t) = H(t,i/N)$ and $H$ is a fixed smooth function of $t$ and $x$. For large $N$, $E_i(t) \approx N^{-1}\partial_xH(t,i/N)$.

Following similar steps as in the previous section for BEP($m$) but now with the generator
\begin{equation*}
L^{\mathrm{WABEP}(m)}_N=L_N^{\mathrm{BEP}(m)}+ \Gamma,
\end{equation*}
we have
\begin{equation*}
 \Gamma \langle\pi_N(Z(t)),\phi \rangle
 =N \sum_{i=1}^{N}  \phi_i \Gamma z_i,
\end{equation*}
and from the choice of $\Gamma$ above we obtain
\begin{align*}
\Gamma z_i &=-N^2\Bigl[ E_{i-1} z_{i-1} z_{i} \left( \partial_{i-1} - \partial_{i} \right) + E_i z_i z_{i+1} \left( \partial_i - \partial_{i+1} \right)\Bigr] z_i\\
&=N^2\bigl[E_{i-1} z_{i-1} z_{i}- E_i z_i z_{i+1}\bigr],
\end{align*}
so that
\begin{equation*}
\Gamma \langle\pi_N(Z(t)),\phi \rangle
 =N \sum_{i=1}^{N}  \phi_i [ E_{i-1} z_{i-1} z_{i}- E_i z_i z_{i+1}].
\end{equation*}

Again applying summation by parts, we get
\begin{equation*}
\Gamma \langle\pi_N(Z(t)),\phi \rangle  =N \sum_{i=1}^{N}  (\phi_{i+1}-\phi_i) E_i z_i z_{i+1}.
\end{equation*}
Combining this with the computation in the previous section we find
\begin{align*}
L^{\mathrm{WABEP}(m)}_N \langle\pi_N(Z(t)),\phi \rangle
&=mN \sum_{i=1}^{N}  z_i \big[ \phi_{i-1}+\phi_{i+1}-2 \phi_i \big] + N \sum_{i=1}^{N}  (\phi_{i+1}-\phi_i) E_i z_i z_{i+1}
\\
& =\frac{m}{N} \sum_{i=1}^{N}  z_i \Delta_N \phi_i + \frac{1}{N} \sum_{i=1}^{N}  z_i z_{i+1} \nabla_N H_i \nabla_N \phi_i
\\
&\approx  m\langle \pi_N(Z(t)),\partial_{xx} \phi\rangle + \langle \left[ \pi_N ^2 (Z(t)) \right]  \partial_x H ,\partial_x \phi\rangle,
\end{align*}
where we write $\nabla_N H_i = N(H_{i+1}-H_i)$ for the discrete derivative and we define
\begin{equation*}
 \pi_N^{(2)}(Z(t))  :=\frac{1}{N} \sum_{i=1}^{N}  z_i z_{i+1} \delta_{{i}/{N}} .
\end{equation*}
Note that $\pi_N^{(2)}$ is not simply the square of $\pi_N$; however,
in the limit $N\to\infty$ we expect that $\wlim_{N\to\infty} \pi_N^{(2)}(Z) = \bigl(\wlim_{N\to\infty}\pi_N(Z)\bigr)^2$, as we will see in the next section. This property is crucial in order to obtain a closed equation for $\rho(t,x)$, and is related to the fact that in the hydrodynamic time scale $N^2$, local equilibrium installs.

\section{The Replacement Lemma}
\label{sec:replacement}
We now discuss the Replacement Lemma for the WABEP$(m)$ in more detail.
The term
\be\label{problem}
\frac{1}{N} \sum_{i=1}^{N}  z_i z_{i+1} \nabla_N H_i \nabla_N \phi_i
\ee
is not a function of $\pi_N(Z)$, and in order to close the equation
we need to replace this term by a function of $\pi_N(Z)$. This is a classical
problem in the theory of hydrodynamic limits, and various methods have been devised to tackle it~\cite{GPV,Yau91,GrunewaldOttoVillaniWestdickenberg09}. In the context of a Ginzburg-Landau model Guo, Papanicolaou, and Varadhan \cite{GPV,KOV} show that a term of the form~\eqref{problem} can
be replaced by a function of $\pi_N$. More precisely, for $\psi:\T\to\R$ smooth and any $\epsi>0$,

\beq\label{approx}
\frac{1}{N} \sum_{i=1}^{N}  z_i (t) z_{i+1} (t) \psi_i
&= &
\frac{1}{N} \sum_{i=1}^{N}  \frac{1}{2\epsi N+1}\sum_{|j-i|\leq \epsi N}z_j (t) z_{j+1} (t) \psi_i + O(\epsi)
\nonumber\\
&= &
\frac{1}{N} \sum_{i=1}^{N}  \nu_{\frac{1}{2\epsi N+1}\sum_{|j-i|\leq \epsi N}z_j (t)} (z_0z_1) \psi_i + O(\epsi)
\nonumber\\
&= &
\frac{1}{N} \sum_{i=1}^{N}  \biggl(\frac{1}{2\epsi N+1}\sum_{|j-i|\leq \epsi N}z_j (t)\biggr)^2 \psi_i + O(\epsi).
\eeq

The first equality follows from the smoothness of $\psi$.
In the second step $\nu_\rho$ is the product measure \eqref{def:nu} on the two-site state space with variables $z_0$ and $z_1$, with $\E_{\nu_\rho} (z_0) = \E_{\nu_\rho}(z_1) =\rho$,  $\nu_\rho(f)=\int f d\nu_\rho$,
and
hence $\nu_\rho (z_0z_1)= \rho^2$. Given this definition, the third equality follows from these properties of $\nu_\rho$. The main approximation step is therefore the second equality.
Once one has justified this approximation, the expression in the last line
is again a function of the empirical energy profile.

\drop{Indeed denote
by $\tilde{f}(\rho)=\nu_\rho(f)$. Next denote $\ell_\epsi (x) =\frac1{2\epsi} 1_{[-\epsi, \epsi]} (x)$, then
\[
\frac1{2N\epsi} \sum_{|j|\leq \epsi N} z_j = \langle\pi_N(Z), \ell_\epsi\rangle,
\]
and therefore, choosing $f(z)= z_0z_1$ and correspondingly $\tilde{f}(\rho)=\rho^2$ we have
\begin{eqnarray*}
&&\frac{1}{N} \sum_{i=1}^{N}  \nu_{\frac{1}{2\epsi N+1}\sum_{|j-i|\leq \epsi N}z_j (N^2 t)} (z_0z_1) \psi_i
\\
&&
\frac{1}{N} \sum_{i=1}^{N}  \theta_i \tilde{f} (\langle\pi^N(Z(N^2 t)), \ell_\epsi\rangle) \psi_i
\\
&&\approx\int \psi(x) \tilde{f}(\rho(t,x)) dx
\\
&&\approx\int \psi(x) (\rho(t,x))^2 dx
\end{eqnarray*}
}

Intuitively, the approximation in the second equality arises from the local-equilibrium property.
However, since one wants to apply it in the context of large deviations, one
has to show that the approximation is superexponentially good.
This means that if we denote for given $\delta>0$
the event
\begin{equation*}
\left|\int_0^T\left(\frac{1}{N} \sum_{i=1}^{N}  \frac{1}{2\epsi N+1}\sum_{|j-i|\leq \epsi N}z_j (N^2t) z_{j+1} (N^2t) \psi_i
-
\frac{1}{N} \sum_{i=1}^{N}  \nu_{\frac{1}{2\epsi N+1}\sum_{|j-i|\leq \epsi N}z_j (N^2t)} (z_0z_1) \psi_i \right)ds\right| >\delta
\end{equation*}
by $A^N_\delta$, then
\be\label{super}
\limsup_{N\to\infty}\frac1N\log \pee^{\mathrm{BEP}(m)}_N (A^N_\delta)=-\infty.
\ee
For the Ginzburg-Landau model of~\onlinecite{GPV} the  superexponential estimate and   the dynamical large-deviation principle was proven in \onlinecite{DV}.
See \onlinecite{KipnisLandim99} for more details and a proof in the context of  gradient zero range processes. To apply the same method to the WABEP there are two additional technical difficulties, similar to the case of the KMP model~\cite{BertiniGabrielliLebowitz05}, namely first that the
equilibrium invariant measures have only exponential tails,  and second that the mobility is equal to $\rho^2$ and thus is unbounded.

\medskip

Assuming that the Replacement Lemma holds, we can proceed with the hydrodynamic limit of
the WABEP$(m)$, and we obtain from the computation in the two previous sections
\begin{equation*}
L^{\mathrm{WABEP}(m)}_N \langle\pi_N(Z(t)),\phi \rangle
\approx  m\int \rho(t,x) \partial_{xx} \phi (x) +  \int \rho(t,x)^2 \partial_x H(t,x) \partial_x \phi(x),
\end{equation*}
which leads to the hydrodynamic limit given by the equation
\be\label{wabeplim}
\partial_t \rho= m\,\partial_{xx} \rho- \partial_x\left(\rho^2\partial_x H\right).
\ee

\section{Large deviations from the hydrodynamic limit}
\label{sec:LDP}
In the language of the general large-deviation theory developed by
\onlinecite{BertiniDe-SoleGabrielliJona-LasinioLandim09}, equation~\eqref{eq:limitBEP} is a model with diffusion constant $D=m$
and
susceptibility (mobility) $\chi(\rho)= \rho^2$. Assuming the validity of  the superexponential Replacement Lemma,
and as a consequence the limiting equation for the WABEP$(m)$, the road to large deviations is
well known.
We describe this procedure here somewhat intuitively.
The hydrodynamic limit for the BEP$(m)$ predicts that
$\pi_N(Z(t))$ evolves ``typically'' as $\rho(t,x) dx$ where
$\rho(t,x)$ solves the diffusion equation~\eqref{eq:limitBEP}. This behavior
is a manifestation of the law of large numbers and therefore one expects corresponding
exponentially small (in $N$) large-deviation probabilities, i.e.,
we expect, in the sense of large deviations
\be
\pee_N^{\mathbb{BEP}(m)} (\pi_N(Z)|_{[0,T]} \approx \gamma dx|_{[0,T]}) \sim e^{-N I(\gamma)}.
\ee
For the rate function $I$ one expects the Lagrangian form
\[
I(\gamma)= \int_0^T\int_\T \loc(\gamma(t,x),\dot\gamma(t,x)) \,dx dt.
\]
In order to find $\loc$, one modifies the BEP  by adding a weak time-dependent asymmetry $H(t,x)$ as we described above,
in such a way that the hydrodynamic equation
\eqref{wabeplim}
has the trajectory $\gamma$ as a solution, i.e., we look for $H= H(t,x)$ such that
\be\label{wabeplim2}
\partial_t \gamma(t,x)= m\partial_{xx} (\gamma(t,x))- \partial_x\left(\gamma(t,x)^2\partial_x H(t,x)\right).
\ee
Here $H(t,x)$ is the unknown function, i.e., this is a linear Cauchy problem in $H$ for given $\gamma$.
The solution $H$ can be written formally as
\be
\partial_x H(t,x)= \frac{1}{\gamma(t,x)^2}\partial_x^{-1} \left(\partial_t \gamma(t,x)- m\partial_{xx} (\gamma(t,x))\right)
\ee
With this choice of $H$ the event $\pi_N(Z) \approx \gamma$
becomes ``typical'', i.e., of probability close to one as $N\to\infty$.
Therefore, one can use the  ``Cram\'er'' tilting method:
\begin{align*}
\log\pee^{\mathrm{BEP}(m)}_N (\pi_N(Z) \approx \gamma dx)
&=\log\int \Indicator(\pi_N \approx \gamma dx)\frac{d\pee^{\mathrm{BEP}(m)}_N}{d\pee^{\mathrm{WABEP}(m)}_{N, H}} \;d\pee^{\mathrm{WABEP}(m)}_{N,H}
\\
&
\geq
\int  \Indicator(\pi_N \approx \gamma dx) \log\frac{d\pee^{\mathrm{BEP}(m)}_N}{d\pee^{\mathrm{WABEP}(m)}_{N, H}}\;d\pee^{\mathrm{WABEP}(m)}_{N,H}
\end{align*}
where $\Indicator$ is the indicator function.
So one has to compute
\[
\log\frac{d\pee^{\mathrm{BEP}(m)}_N}{d\pee^{\mathrm{WABEP}(m)}_{N, H}}.
\]
The Radon-Nikodym derivative ${d\pee^{\mathrm{BEP}(m)}_N}/{d\pee^{\mathrm{WABEP}(m)}_{N, H}}$ is given by the Girsanov formula (see the appendix):
\be\label{girsbol}
\frac1N \log\frac{d\pee^{\mathrm{BEP}(m)}_N}{d\pee^{\mathrm{WABEP}(m)}_{N, H}}= -\frac18 \int_0^T \sum_{i} E_i(t) z_i z_{i+1} + \mathcal M
\ee
with $E_i(t) =\partial_x H(i/N, t)$. Here $\mathcal M$ is a martingale under $\pee^{\mathrm{WABEP}(m)}_{N,H}$. Using the superexponential estimate \eqref{super} this is rewritten
in terms of the density field, which as in the derivation of the hydrodynamic limit leads to
\[
\frac{1}{N}\int   \Indicator(\pi_N \approx \gamma dx) \log\frac{d\pee^{\mathrm{BEP}(m)}_N}{d\pee^{\mathrm{WABEP}(m)}_{N, H}} \;d\pee^{\mathrm{WABEP}(m)}_{N, H}
\approx
-\frac18 \int_0^T\!\! \int   (\partial_x H(x,t))^2 \gamma(t,x)^2 \, dxdt.
\]

Therefore the large deviations of the time courses of $\pi_N(Z)$ are described by
\[
\mathbb P_N^{\mathrm{BEP}(m)} \bigl( \pi_N(Z)|_{[0,T]}| \approx \rho|_{[0,T]}\bigr)
\sim \exp\bigl\{ -NI^{\mathrm{BEP}(m)}_{\rho^0}(\rho)\bigr\},
\]
where the rate function $I^{\mathrm{BEP}(m)}_{\rho^0}$ is
\begin{equation*}
\label{def:IBEPm}
I^{\mathrm{BEP}(m)}_{\rho^0}(\rho) =\begin{cases}
\displaystyle \frac18\int_0^T\!\!\!\int_\T  \rho(t,x)^2 (\partial_x H (t,x))^2\, dxdt,\quad
& \text{if $\rho(0,\cdot) = \rho^0$ and  $H$ is given by~\eqref{wabeplim}},\\[2\jot]
+\infty &\text{otherwise}.
\end{cases}
\end{equation*}
Note that $I^{\mathrm{BEP}(m)}_{\rho^0}$ is a measure of the degree to which $\rho$ does \emph{not} satisfy the equation $\partial_t\rho = m\partial_{xx}\rho$.

\section{The GBEP($a$) and KMP processes}
\label{sec:otherprocesses}
The calculations above were done for the Brownian Energy Process with parameter $m$. For the GBEP($a$) and KMP processes the arguments are very similar. In fact, the invariant measures of the GBEP($a$) are the same as for BEP(1), and for the KMP the same as for BEP(2). This also implies that the large-deviation rate functionals for the two processes are $S_{\rho_0,1}$ and $S_{\rho_0,2}$.

Therefore  the three processes have equilibrium rate functionals $S_{\rho_0,m}$, $S_{\rho_0,1}$, and $S_{\rho_0,2}$ that are essentially the same; they only differ by a multiplicative constant.

\medskip

For the GBEP($a$), the addition of the factor $a$ in~\eqref{def:GBMP} and~\eqref{def:GBEP} modifies the hydrodynamic limit, which now becomes
\[
\partial_t \rho = \partial_{x}(a^2(\rho)\partial_x \rho).
\]
For the large deviations of the time courses of the empirical measure $\pi_N$, similar arguments as for BEP($m$) give the expression
\begin{equation}
\label{def:IGBEP}
I^{\mathrm{GBEP}(a)}_{\rho^0}(\rho) = \frac18 \int_0^T
\!\!\int_\T  \rho^2 a^2(\rho) (\partial_x H)^2,
\qquad \partial_t \rho - \partial_{x}(a^2(\rho)\partial_x \rho) = \partial_x\bigl(\rho^2a^2(\rho)\partial_x H\bigr),
\end{equation}
provided that $\rho(0,\cdot) = \rho^0$.

For the KMP, the hydrodynamic limit and the large deviations were found in~\onlinecite{BertiniGabrielliLebowitz05} to be
\[
\partial_t \rho = \partial_{xx} \rho
\]
and
\begin{equation}
\label{def:IKMP}
I^{\mathrm{KMP}}_{\rho^0}(\rho) = \frac14\int_0^T
\!\!\int_\T  \rho^2 (\partial_x H)^2,
\qquad \partial_t \rho - \partial_{xx} \rho = \partial_x(\rho^2\partial_x H),
\end{equation}
again provided that $\rho(0,\cdot) = \rho^0$.

\section{Gradient-flow structures}
\label{sec:GF}

We now describe how these results identify gradient-flow structures for the limiting equations. Above we have derived three large-deviation rate functionals $I^{\mathrm{BEP}(m)}_{\rho^0}$, $I^{\mathrm{GBEP}(a)}_{\rho^0}$, and $I^{\mathrm{KMP}}_{\rho^0}$, which each can be written in the form
\begin{equation}
\label{def:I}
I(\rho) := \frac12 \int_0^T \left\|\partial_t \rho-K_\rho \frac{\delta E}{\delta \rho}\right\|_{K_\rho^{-1}}^2\, dt,
\qquad \text{provided }\rho(0,\cdot) =\rho^0,
\end{equation}
where
\begin{equation}
\label{def:EK}
E(\rho) := -\gamma \int_\T \log \rho \qquad\text{and}\qquad
K_\rho\xi := -\partial_x \bigl(\alpha(\rho)\partial_x\xi\bigr).
\end{equation}
Here $\gamma$ and $\alpha$ are process-dependent. We list their values and the corresponding limiting equation:
\begin{subequations}
\label{def:GFstructures}
\begin{alignat}4
&\mathrm{BEP}(m): &\quad &\gamma = m/4 &\quad \text{and }&\alpha(\rho) =4\rho^2;&& \partial_t \rho = m\partial_{xx}\rho;\\
&\mathrm{GBEP}(a): &&\gamma = 1/4 & \text{and }& \alpha(\rho) = 4\rho^2a(\rho)^2;
 &\qquad& \partial_t \rho = \partial_{x}\bigl(a^2(\rho)\partial_x\rho\bigr);
 \label{GF-structure:GBEP}\\
&\mathrm{KMP}: &&\gamma =1/2&\text{and } & \alpha(\rho) = 2\rho^2;&&
\partial_t \rho = \partial_{xx}\rho.
\end{alignat}
\end{subequations}
Note that for each of the processes $E$ is equal to one-half of the corresponding equilibrium rate functional $S_{\rho_0,m}$, $S_{\rho_0,1}$, or $S_{\rho_0,2}$ up to an additive constant, since $\int\rho$ can be assumed constant during evolution.

The norm $\|\cdot\|_{K_\rho^{-1}}$ in~\eqref{def:I} is defined as the dual norm of
\[
\|\xi\|^2_{K_\rho} := \int_\T \xi K_\rho \xi,
\]
Here, for a norm $\| \|$ on a linear space $X$, the dual norm $\|\|_*$ on the dual space $X'$ is defined via
\[
\|f\|^2_* = \sup \frac{|\langle f,g\rangle|^2}{\|g\|^2}
\]
where $\langle f,g\rangle$ denotes the pairing between $X'$ and $X$.
It can be written in terms of a formal inverse $G_\rho = K_\rho^{-1}$ as
\[
\|s\|^2_{K_\rho^{-1}} = \|s\|^2_{G_\rho} := \int_\T s G_\rho s.
\]
To formally verify the claim above that $I^{\mathrm{BEP}(m)}_{\rho^0}$ can be written in the form~\eqref{def:I}, for instance, first note that $K_\rho \delta E/\delta \rho = m\partial_{xx}\rho$. Then, whenever $\partial_t \rho - m\partial_{xx} \rho = \partial_x\bigl(\rho^2\partial_x H\bigr)$,
\[
K_\rho H = -4 \bigl(\partial_t\rho - m\partial_{xx}\rho\bigr) \qquad\text{so that}\qquad
H = -4G_\rho \bigl(\partial_t\rho - m\partial_{xx}\rho\bigr).
\]
Therefore
\begin{align*}
I^{\mathrm{BEP}(m)}_{\rho^0}(\rho) &= \frac18 \int_0^T \int_\T \rho^2 (\partial_x H)^2
  = -\frac18 \int_0^T \int_\T H \partial_x (\rho^2 \partial_x H)\\
&= \frac1{32} \int_0^T \int_\T H K_\rho H =
  \frac12 \int_0^T\int_\T \bigl(\partial_t\rho - m\partial_{xx}\rho\bigr) G_\rho \bigl(\partial_t\rho - m\partial_{xx}\rho\bigr) = I(\rho).
\end{align*}

\medskip

The structure~\eqref{def:I} identifies a gradient-flow structure with state space $L_{\geq0}^1(\T)$, with driving functional $E$, and with Onsager operator $K_\rho$. This is the structure that we are looking for.

\section{Conclusions and discussion}
\label{sec:discussion}

In the Introduction we asked the question `what is the natural gradient-flow structure for the conduction of heat?' By considering three different microscopic stochastic processes that model heat conduction we identified, through their large deviations, three gradient structures for macroscopic heat equations. We now comment on our findings.

\subsection*{The driving functional $E$}
Each of the three gradient-flow structures~\eqref{def:GFstructures} is driven by a functional $E$, which  is the same, up to multiplicative and additive constants, as the large-deviation rate functional $S_{\rho,\cdot}$ of the equilibrium invariant measure of the process. Therefore, despite the differences in the \emph{dynamics} of the processes, the functional that \emph{drives} those dynamics is nearly the same for the three processes. This is of course a consequence of the similarity of the invariant measures themselves, and therefore it is interesting to understand these similarities.

Although it is difficult to pin down necessary and sufficient conditions leading to the invariant measures of these processes, one can identify some ingredients:
\begin{itemize}
\item The variable $p_i$ (or the variables $p_i^j$) represents a quantity on a \emph{fixed} lattice;
\item The evolution locally preserves the quantity $p_i^2+p_{i+1}^2$ (or $\sum_{j,k=1}^m (p_i^j)^2 + (p_{i+1}^k)^2$);
\end{itemize}
In the thermodynamic limit ($N\to\infty$), using the equivalence of ensembles results in an invariant 
product measure with Lebesgue density $\exp (-e_i/\theta)$,
at each site $i$ 
where $e_i = p_i^2 $ or $e_i = \sum_j (p_i^j)^2$. After transition to the variables $z_i$, this invariant measure transforms into the Gamma distribution $\nu_{N,\theta,m}$ defined in~\eqref{def:nu};  all Gamma distributions lead to the same form of $E$, with $-\gamma \log \rho$ as the important term under the integral.

Summarizing, the form of the driving functional $E$ is related to the occurrence of a conserved quadratic energy $p^2$ on a fixed lattice.

\medskip

There is also a classical thermodynamical argument that leads to a driving functional proportional to the logarithm of the temperature. We start from the postulate for simple closed homogeneous systems without volume change,
\def\dbar{{\mathchar'26\mkern-12mu d}}
\begin{equation}
\label{eq:Gibbs}
T dS = \dbar Q
\end{equation}
where $dS$ is the infinitesimal change in entropy $S$ while adding $\dbar Q$ of heat to a system at temperature $T$ (the barred $d$ is the usual way of indicating that $\dbar Q$ is not necessarily an exact differential (see e.g.~\onlinecite[Ch.~1]{Callen85})). Taking the example of BEP($m$), internal energy is proportional to temperature (see the relationship between $\theta$ and $\E_{\nu_{\theta,m}} [ z_i ]$ in~\eqref{def:rho0}): $U = mT/2$. Therefore $\dbar Q$ is equal to the change $dU$ in internal energy $U$, and \eqref{eq:Gibbs} reduces to
\[
dS = \frac m2\frac1U dU = \frac m2 d\log U.
\]
If we accept the interpretation of the thermodynamic entropy $S$ as (minus) the large-deviation rate functional of the equilibrium system, then this relation provides a separate argument why the driving functional should be proportional to the logarithm of the temperature.

\subsection*{The Onsager operator $K$}
While $S_{\rho,\cdot}$ and therefore $E$ is generated by the \emph{conservation} properties of the process, the rates in the generator determine the `diffusion constant' $\alpha(\rho)$ in $K_\rho$~\eqref{def:EK}. The example of the GBEP($a$)~\eqref{GF-structure:GBEP} shows that there is nearly full freedom here: $\alpha(\rho) = a^2(\rho)$ can be any positive function of $\rho$, and under certain conditions this coefficient can even vanish at certain values of $\rho$ (`degenerate diffusion'~\onlinecite{Peletier97TH}).

A special feature of the BEP($m$) and KMP models is that the mobility is quadratic ($\alpha(\rho) \sim \rho^2$) and the driving functional logarithmic; the two combine to provide a \emph{linear} heat equation, through $\partial_x (\rho^2 \partial_x(-1/\rho)) = \partial_{xx}\rho$. It is interesting to ask where this quadratic mobility comes from.

In the BEP($m$) the quadratic mobility can be traced back to two ingredients:
\begin{enumerate}
\item The stochastic force that lattice point $i+1$ exerts on its neighbor at $i$ scales as $p_{i+1}$ (see the first two terms in~\eqref{eq:sde-pi}), and therefore its effect in $L^{\mathrm{BMP}}$ as $p_{i+1}^2$~(see~\eqref{def:L-BMP}).
\item Since $dp_i^2/dt= 2p_i dp_i/dt$, the effect of this forcing on $z_i$ is multiplied by $p_{i}$ in the SDE~\eqref{eq:sde-zi} and by $p_{i}^2$ in $L^{\mathrm{BEP}(m)}$ (see~\eqref{def:L-BEPm}).
\end{enumerate}
These are the reasons for the mobility $z_iz_{i+1}$ that multiplies the second-order term in~\eqref{def:L-BEPm}.

For the KMP model, there is an additional explanation. Since transfer of energy happens by redistributing energy with a  uniform distribution, the expected transfer of energy scales linearly with the difference of the expected energies in the two lattice points. In formulas,
\[
L_{i,i+1}^{\mathrm{KMP}}z_i  = \int_0^1 \bigl[ s(z_i+z_{i+1}) -z_i\bigr] \, ds = \frac12 (z_i+z_{i+1}).
\]
This linear energy flux translates into a linear hydrodynamic limit.

\subsection*{Comparison with diffusion}
Comparing these systems with diffusion systems we observe some similarities and some differences.

\emph{Similarity 1:} Both concern the redistribution over time of a conserved quantity: energy for heat conduction and mass for diffusion. Consequently both have hydrodynamic limits in divergence form.

\emph{Difference 1:} The stationary rate functionals are different: $-\int\log \rho$ for heat conduction, versus $\int \rho\log\rho$ for diffusion. These differences reflect the different origins: $-\int\log \rho$ arises from a Gamma-distributed quantity \emph{at} each lattice point, while $\int\rho\log\rho$ arises from the \emph{moving-around} of masses, through Stirling's formula (see e.g.~\onlinecite[Sec.~4.2]{PeletierVarMod14TR}). 

\emph{Difference 2:} The mobilities are different: $\rho^2$ for KMP and BEP($m$), $\rho^2a^2(\rho)$ for GBEP($a$), whereas diffusion has mobility $\rho$. Again this reflects the different origins: as described above, the mobility $\rho^2$ for KMP and BEP($m$) arises from the combination of neighboring forces that scale as $p$ with energies that scale as $p^2$, while the mobility $\rho$ for diffusion arises from a simple counting argument, where the flux of particles is proportional to the number of particles.

In this context it is instructive also to compare to the Simple Symmetric Exclusion Process, the jump process for particles on a lattice in which particles jump left and right with rate $1$ whenever the destination site is empty. Here the mobility is $\rho(1-\rho)$, which can be constructed through the same counting argument as above (the factor $\rho$) mitigated by the probability of finding the neighboring site empty (the factor $(1-\rho)$).

\subsection*{No metric space}
When the mobility $\alpha(\rho)$ in $K_\rho$ is concave, a generalized Wasserstein-type metric can be constructed as an infimum along curves, following~\onlinecite{BenamouBrenier00}:
\begin{equation}
\label{def:d-BB}
d(\rho_0,\rho_1)^2 := \inf_{\rho,w}\left\{\int_0^1\int_\T \frac{w^2}{\alpha(\rho)}: \partial_t \rho = \partial_x w\text{ and } \rho|_{t=0,1} = \rho_{0,1}\right\}.
\end{equation}
The concaveness of $\alpha$ causes the function $(w,r)\mapsto w^2/\alpha(r)$ to be convex, rendering the integral above lower-semicontinuous with respect to various weak topologies for $w$ and $\rho$. This property is used in~\onlinecite{CarrilloLisiniSavareSlepcev10} to construct gradient-flow solutions with respect to this metric, in the sense of~\onlinecite{AmbrosioGigliSavare05}, obtaining various other properties in the process.

However, if $\alpha$ is \emph{not} concave, then this metric need not be well-defined; for instance, if $\lim_{r\to\infty} \alpha(r)/r =+\infty$, then it is not difficult to see that for each $\rho_0,\rho_1$ with $\int(\rho_0-\rho_1)=0$ the infimum in~\eqref{def:d-BB} is zero\cite{Note2}.
The important case $\alpha(\rho) = \rho^2$ (BEP($m$) and KMP) is an example of this; the natural metric~\eqref{def:d-BB} is therefore not well defined for these systems.

\subsection*{Generalizations to other models}
Although the fundamental question of this paper---what is the natural gradient structure for heat conduction?---is most simply posed for the simplest of heat-conducting systems, the insights that we gained are also useful for more complex models. For instance, in the `variational-modeling' philosophy~\cite{PeletierVarMod14TR} one constructs mathematical models of real-world systems by choosing a driving functional and an Onsager operator. These two choices fully determine the system. For this methodology to function well, one needs to understand how to choose these components---and this is exactly the question of this paper.

\subsection*{Acknowledgements} 
The authors are grateful for many interesting discussions with Jin Feng. 
MAP and KV acknowledge the support of NWO VICI grant 639.033.008.

\appendix

\section{A Girsanov Formula}
\label{app:Girsanov}

Here we prove \eqref{girsbol}. We refer to \onlinecite{stroockvaradhan} for more details on the
proof of the Girsanov formula in the context of diffusion processes; the formula is proved for bounded coefficients in~\onlinecite[Th.~6.4.2]{stroockvaradhan}, and the extension to unbounded coefficients can be done with the methods of~\onlinecite[Ch.~11]{stroockvaradhan}.

Consider a Markov process with generator
\begin{equation}
L=\frac12 \sum_i \alpha^i(x)\partial_i^2+\sum_i \beta^i(x)\partial_i
\end{equation}
and another Markov process obtained from $L$ by addition of an extra drift term
of the form
\begin{equation}
\hat {L} =\frac12 \sum_i \alpha^i(x)\partial_i^2+\sum_i \beta^i(x)\partial_i + \sum_i \alpha^i(x) \gamma^i_t(x)\partial_i.
\end{equation}
These two processes correspond to the following stochastic differential equations:
\begin{equation}
d x_t^i = \beta^i(x_t) \,dt + \sqrt{ \alpha^i(x_t) } \,dB_t^i
\end{equation}
and
\begin{equation}
d \hat x_t^i = \left(\beta^i(x_t) + \alpha^i(x_t)\gamma^i_t(x_t)\right) dt + \sqrt{ \alpha^i(x_t) }\, dB_t^i.
\end{equation}
Let us call $\pee$, resp.\ $\hat{\pee}$ the path space measure of the $x$, resp.\ $\hat{x}$ process.
The Girsanov formula then gives the Radon Nikodym derivative of the $\hat{x}$ process with respect to
the $x$ process:
\be\label{girsanooo}
\frac{d\hat\pee}{d\pee} = \exp\left(\sum_i \left( \int_0^T  \gamma_t^i(x_t)( dx_t^i - \beta^i(x_t)\,dt) - \frac{1}{2} \int_0^T (\gamma_t^i)^2(x_t)\alpha^i(x_t)\, dt   \right)\right).
\ee
As a consequence,
in the case  $\gamma_t^i = \gamma (\frac{i}{N}, t)$, we find
\begin{eqnarray}
\frac{1}{N} \log \frac{d \hat {\mathbb P}}{d \mathbb P}
&=&\frac{1}{N} \sum_i \left( \int_0^T  \gamma_t^i(x_t)( dx_t^i - \beta^i(x_t)\,dt) - \frac{1}{2} \int_0^T (\gamma_t^i)^2(x_t)\alpha^i(x_t)\, dt  \right)
\nonumber\\
&=& \frac{1}{2N} \sum_i  \int_0^T (\gamma_t^i)^2(x_t)\alpha^i(x_t)\, dt + \mathcal M,
\end{eqnarray}
where
\begin{equation*}
\mathcal M := \frac{1}{N}\sum_i \int_0^T  \gamma_t^i(x_t)\bigl[ dx_t^i - \beta^i(x_t)\,dt - \alpha^i(x_t)\gamma^i_t(x_t)\, dt\bigr]
\end{equation*}
is a martingale under $\hat \pee$.
We can now apply this to the case of Radon-Nikodym derivative of the BEP$(m)$ with respect to the WABEP$(m)$:
\begin{equation*}
L_N^{\mathrm{BEP}(m)} =  \sum_i 2z_i z_{i+1}
\left(\partial_i-\partial_{i+1}\right)^2 - m (z_i-z_{i+1})\left(\partial_i-\partial_{i+1}\right)
\end{equation*}
and
\begin{equation*}
L_{N,H}^{\mathrm{WABEP}(m)} =  \sum_i 2z_i z_{i+1}
\left(\partial_i-\partial_{i+1}\right)^2 - m (z_i-z_{i+1})\left(\partial_i-\partial_{i+1}\right)
- E_i z_i z_{i+1} \left( \partial_i - \partial_{i+1} \right).
\end{equation*}
We make the change of variables to $y_i:=z_i-z_{i+1}$, and defining $\alpha_i(y)=4z_i z_{i+1}$, $\beta_i(y)=-my_i$ and $\gamma_i(y)=-\frac14E_i=-\frac14E(\frac{i}{N}, u) $
we arrive at the same situation as before to obtain
\begin{equation*}
\frac1N\log \frac{d \mathbb P_{N,H}^{\mathrm{WABEP}(m)}}{d  {\mathbb P_N^{\mathrm{BEP}(m)}}} = \frac18  \int_0^T \sum_i  E_i^2  z_i z_{i+1}  dt + \mathcal M,
\end{equation*}
where $\mathcal M$ is a martingale under $\pee^{\mathrm{WABEP}(m)}_{N,H}$.

\end{document}